\documentclass[conference]{IEEEtran}
\ifCLASSINFOpdf
  \usepackage[pdftex]{graphicx}
  \DeclareGraphicsExtensions{.pdf,.jpeg,.png,.eps}
\else
  \usepackage[dvips]{graphicx}
  \DeclareGraphicsExtensions{.eps}
\fi
\usepackage[cmex10]{amsmath}
\usepackage{amssymb,amsthm}
\usepackage{mathtools, cuted}
\usepackage{lipsum, color}
\usepackage[utf8]{inputenc}
\usepackage{bbm}
\usepackage{cite}
\usepackage[left=0.55in,right=0.55in,bottom=0.675in,top=0.675in]{geometry}
\usepackage[tight,footnotesize]{subfigure}
\newcommand{\argmin}{\operatornamewithlimits{arg\,min}}
\newcommand{\argmax}{\operatornamewithlimits{arg\,max}}
\newcommand{\abs}[1]{\left\lvert{#1}\right\rvert}
\newcommand{\norm}[1]{\left\lVert{#1}\right\rVert}

\newcommand{\nn}{\nonumber}

\newtheorem{theorem}{Theorem}
\newtheorem{prop}{Proposition}

\newtheorem{cor}{Corollary}

\begin{document}
\title{On The Compound MIMO Wiretap Channel with Mean Feedback}
% \author{Michael~Shell,
%         John~Doe,~\IEEEmembership{Fellow,~OSA,}
%         and~Jane~Doe,~\IEEEmembership{Life~Fellow,~IEEE}% <-this % stops
% qa space
\author{\IEEEauthorblockN{$^{\dagger}$Amr Abdelaziz, $^{\dagger}$C. Emre Koksal and $^{\dagger}$Hesham El Gamal & $^*$Ashraf D. Elbayoumy} \IEEEauthorblockA{$^{\dagger}$Department of Electrical and Computer Engineering & $^*$Department of Electrical Engineering\\
The Ohio State University & Military Technical College\\
Columbus, Ohio 43201 & Cairo, Egypt
}}
\maketitle

\begin{abstract}
Compound MIMO wiretap channel with double sided uncertainty is considered under channel mean information model. In mean information model, channel variations are centered around its mean value which is fed back to the transmitter. We show that the worst case main channel is anti-parallel to the channel mean information resulting in an overall unit rank channel. Further, the worst eavesdropper channel is shown to be isotropic around its mean information. Accordingly, we provide the capacity achieving beamforming direction. We show that the saddle point property holds under mean information model, and thus, compound secrecy capacity equals to the worst case capacity over the class of uncertainty. Moreover, capacity achieving beamforming direction is found to require matrix inversion, thus, we derive the null steering (NS) beamforming as an alternative suboptimal solution that does not require matrix inversion. NS beamformer is in the direction orthogonal to the eavesdropper mean channel that maintains the maximum possible gain in mean main channel direction. Extensive computer simulation reveals that NS performs very close to the optimal solution. It also verifies that, NS beamforming outperforms both maximum ratio transmission (MRT) and zero forcing (ZF) beamforming approaches over the entire SNR range. Finally, An equivalence relation with MIMO wiretap channel in Rician fading environment is established.  
\end{abstract}


% \fi
\IEEEpeerreviewmaketitle

\section{Introduction}
A key consideration in determining the secrecy capacity of the MIMO wiretap channel is the amount of information available at the transmitter, not only about the eavesdropper channel, but also about the main channel. In principle, assuming perfect knowledge about the main channel, either full eavesdropper's channel state information (CSI) or, at least, its distribution are required to determine the secrecy capacity. The secrecy capacity of the general MIMO wiretap channel has been studied in \cite{elder_wiretap,MIMOME,mimo_sec_cap} assuming perfect knowledge of both channels.
\par
In practical scenarios, having even partial knowledge on the eavesdropper channel is typically not possible, especially, when dealing with strictly passive eavesdroppers. Further, in fast fading channel, it may also be unreasonable to have perfect main CSI at the transmitter. \textit{Compound wiretap channel} \cite{blackwell1959capacity,wolfowitz1959simultaneous} is a model that tackle these limitations in which CSI is known only to belong to a certain class of uncertainty. This assumption can be used to model eavesdropper only CSI \cite{liang2009compound} (single sided uncertainty) or both main and eavesdropper channels (double sided uncertainty)\cite{ekrem2010gaussian,schaefer2015secrecy}. Depending on the considered class of uncertainty, the secrecy capacity of the compound wiretap channel can be characterized. 
\par
 Classes of uncertainity in the compound wiretap channel can be characterized in two different categories as it pertains to the set that the main and/or eavesdropper channels belong to: 1) Finite state channels 2) Continous set. The discrete memoryless compound wiretap channel with countably finite uncertainty set was studied in \cite{liang2009compound,bjelakovic2013secrecy}. Meanwhile, the corresponding compound Gaussian MIMO wiretap channel with countably finite uncertainty set is analyzed in \cite{liang2009compound}. In both cases, the secrecy capacity is established only for the degraded case (i.e. main channel is stronger in all spatial directions). Meanwhile, the secrecy capacity itself remains unknown for the general indefinite case (i.e. main channel is stronger in subset of the available spatial directions). A closed form solution was obtained either in case of an isotropic eavesdropper \cite{schaefer2015secrecy} or the degraded case in the high SNR regime \cite{loyka2012optimal}. Although the optimal signaling scheme for the non-isotropic non-degraded case still not known to date in general, necessary conditions for optimality were derived in \cite{loyka2012optimal} and \cite{loyka2013further} for the deterministic known channel case. Recently in \cite{schaefer2015secrecy}, the compound Gaussian MIMO wiretap channel was studied under spectral norm constraint (maximum channel gain) and rank constraint for both single and double sided classes of uncertainty without the degradedness assumption.
 \par
In \cite{schaefer2015secrecy} (Theorem 3) it was shown that, the secrecy capacity of the compound MIMO wiretap channel is upper bounded by the worst case capacity over the considered class of uncertainty. The term worst case capacity is established by optimizing the input signal covariance for all possible main and eavesdropper channel, then, taking the minimum over all main and eavesdropper channels over the considered class of uncertainty. Moreover, it was also shown that, the compound secrecy capacity is lower bounded by the capacity of the worst possible main and eavesdropper channels. Here, the saddle point needs to be considered, i.e. $\max \min = \min \max$, where the $\max$ is taken over non-negative definite input covariance matrices subject to an average power constraint and the $\min$ is taken over the classes of channel uncertainty. If the saddle point property holds, the compound capacity is fully characterized and is known to match the worst case one. 
 \par
In this paper, we consider the class of channels with double sided uncertainty under channel mean information model. In mean information model, channel is centered around a mean value which is fed back to the transmitter. An example for the mean information model is the channel with a strong Line-of-Sight (LOS) component, the gain of which is known at the transmitter. While it is unlikely to expect the eavesdropper to share its CSI (even its mean channel) in some scenarios, a secure communication system may be designed in a way that puts physical restrictions on the locations of possible attacker. These physical restrictions can be informative to the transmitter and may enable to achieve better secrecy rates by designing its signaling scheme accordingly. We first establish the worst case secrecy capacity of the compound MIMO wiretap channel under mean information model, then, we show that the saddle point property holds. We show that, the worst case main channel is anti-parallel to the channel mean information resulting in an overall unit rank channel. Further, the worst eavesdropper dropper channel is shown to be isotropic around its mean information. Accordingly, generalized eigenvector beamforming is known to be the optimal signaling strategy \cite{loyka2012optimal}\cite{li2009transmitter}. We show that the saddle point property holds under mean information model, and thus, compound secrecy capacity equals to the worst case capacity over the class of uncertainty. Further, as the generalized eigenvector solution requires matrix inversion, we introduce null steering (NS) beamforming, that is, transmission in the direction orthogonal to the eavesdropper mean channel direction maintaining the maximum possible gain in mean main channel direction, as an alternative suboptimal solution. Extensive computer simulation reveals that NS performs extremely close to the optimal solution. It also verifies that, NS beamforming outperforms both maximum ratio transmission (MRT) and zero forcing (ZF) beamforming approaches over the entire SNR range. Finally, An equivalence relation with MIMO wiretap channel in Rician fading environment is established.

%
% To make it mathematically tractable, we first assume that the eavesdropper mean information is known at the transmitter. For this particular scenario, we show that the saddle point property holds under the mean information model, that is, the compound capacity equals to the worst case capacity over the uncertainty sets. Also, we show that beamforming is the optimal signaling strategy that achieves the compound secrecy capacity. We derive the optimal beamformer which is shown to be the null steering beamforming, i.e., transmitting in the direction orthogonal to the eavesdropper mean information that has the maximum possible gain in the direction of the main channel mean information.
%\par
%A more realistic assumption, is to assume that the channel mean information of the eavesdropper is not exactly known at the transmitter, rather, it is known only to lie within a known set of uncertainty. For this particular case, we show that the saddle point property still holds. We also show that, the worst case eavesdropper is non isotropic due to its mean information. Thus, the worst case eavesdropper is shown to have a mean information whose direction is the closest possible direction to the eavesdropper mean information. More precisely, the worst eavesdropper mean information is the projection of the main channel mean information on the set of uncertainty about the eavesdropper mean information. Again, null steering beamforming in the direction orthogonal to the worst eavesdropper mean channel is shown to be the optimal signaling strategy. 

\section{System Model and Problem Statement}
\label{sec:model}
\subsection{Notations}
In the rest of this paper we use boldface uppercase letters for random matrices, uppercase letters for their realizations, bold face lowercase letters for random vectors and lowercase letters for its realizations. Meanwhile, $(.)^{\dagger}$ denotes conjugate transpose, $\mathbf{I}_N$ denotes identity matrix of size $N$, $\det(.)$ denotes matrix determinant operator and $\mathbf{1}_{m \times n}$ denotes a $m \times n$ matrix of all 1's. 
\label{sec:screcy_capacity}
\subsection{System Model}
We consider the MIMO wiretap channel scenario in which a transmitter $\mathcal{A}$ with $N_a > 1$ antennas amounts to transmit a confidential message massage to a receiver, $\mathcal{B}$, having $N_b > 1$ antennas over an unsecure channel in the presence of a passive adversary, $\mathcal{E}$, equipped with $N_e > 1$ antennas. The discrete baseband equivalent channels for the signal received by each of the legitimate destination, $\mathbf{y}$, and the adversary, $\mathbf{z}$, are as follows:

\begin{align}
\label{eq:inout_security}
\mathbf{y} = \mathbf{H}_b   \mathbf{x} + \mathbf{n}_b, &\;\;\;\;\;\;\; \mathbf{z} = \mathbf{H}_e   \mathbf{x} + \mathbf{n}_e,
\end{align} 
% Transmitter Description ==================================
where  $\mathbf{x} \in \mathbb{C}^{N_a\times 1}$ is the transmitted signal vector constrained by an average power constraint $\mathbb{E}[\mathbf{tr}(\mathbf{x}\mathbf{x}^{\dagger})] \leq P$. Also, ${\mathbf{H}_b} \in \mathbb{C}^{N_b \times N_a}$ and $\mathbf{H}_e  \in \mathbb{C}^{N_e \times N_a}$ are the channel coefficients matrices between message source, destination and adversary respectively. Finally, $\mathbf{n}_b \in \mathbb{C}^{N_b\times 1}$ and  $\mathbf{n}_e \in \mathbb{C}^{N_e\times 1}$ are independent zero mean normalized to unit variance circular symmetric complex random vectors for both destination and adversary channels respectively, where, $\mathbf{n}_b\sim \mathcal{CN}(0,\mathbf{I}_{N_b})$ and $\mathbf{n}_e\sim \mathcal{CN}(0,\mathbf{I}_{N_e})$.

%%%%%%%%%%%%%%%%%%%%%%%%%%%%%%%%%%%%%%%%%%%%%%%%%%%%%%%%%%%%%%%%%%%%%%%%%%%%%%%
\subsection{Problem Statement}

 In this paper, we consider the case where the transmitter does not know the exact realizations of both $\mathbf{H}_b$ and $\mathbf{H}_e$. Rather, it only knows that they both belong to a known compact (closed and bounded) uncertainty sets. Under the considered channel mean feedback model, we define channel uncertainty sets as follows: 
\begin{align}
\label{eq:S1}
\mathcal{S}_b &= \{\mathbf{H}_b:\mathbf{H}_b = \mathbf{H}_{\mu b} + \Delta\mathbf{H}_{b}, \abs{\Delta\mathbf{H}_{b}}_2 \leq \epsilon_b,
\nonumber \\&\;\;\;\;\;\;\;\;\mathbf{H}_{\mu b}= \lambda_{\mu b}^{1/2}v_{b}u_{b}^{\dagger} \} ,             
\end{align}
\begin{align}
\label{eq:S2}
\mathcal{S}_e &= \{\mathbf{H}_e:\mathbf{H}_e = \mathbf{H}_{\mu e} + \Delta\mathbf{H}_{e}, \abs{\Delta\mathbf{H}_{e}}_2 \leq \epsilon_e,
\nonumber \\&\;\;\;\;\;\;\;\;\mathbf{H}_{\mu e}= \lambda_{\mu e}^{1/2}v_{e}u_{e}^{\dagger} , u_e \in \mathcal{U} \} ,             
\end{align}
where, $\mathbf{H}_{\mu \circ}$ is the channel mean information which is assumed to be of unit rank and $ v_{\circ} \in \mathbb{C}^{N_{\circ} \times 1} , u_{\circ} \in \mathbb{C}^{N_a \times 1}$. We assume that the transmitter knows $u_b$, meanwhile, it knows only that $u_e \in \mathcal{U}$ where $\mathcal{U}$ is the set of uncertainty about the eavesdropper mean information. In the extreme case when $u_e$ is know exactly at the transmitter, we simply write $\mathcal{U}=\{u_e\}$.
\par
Further, $\Delta\mathbf{H}_{\circ}$ is the channel uncertain part which is assumed to satisfy the bounded spectral norm condition $\abs{\Delta\mathbf{H}_{\circ}}_2 \leq \epsilon_{\circ}$. In the compound wiretap channel, channels realizations are assumed to be fixed over the entire transmission duration. Therefore, $\Delta\mathbf{H}_{\circ}$ is considered fixed once it has been realized. This model is the scenario in which the eavesdropper can approach the transmitter up to a certain distance and from limited range of directions, see Fig. (\ref{fig:Compound_wiretap_model}).
\begin{figure}[ht]
\centering
\includegraphics[width=3 in, height = 1.5 in]{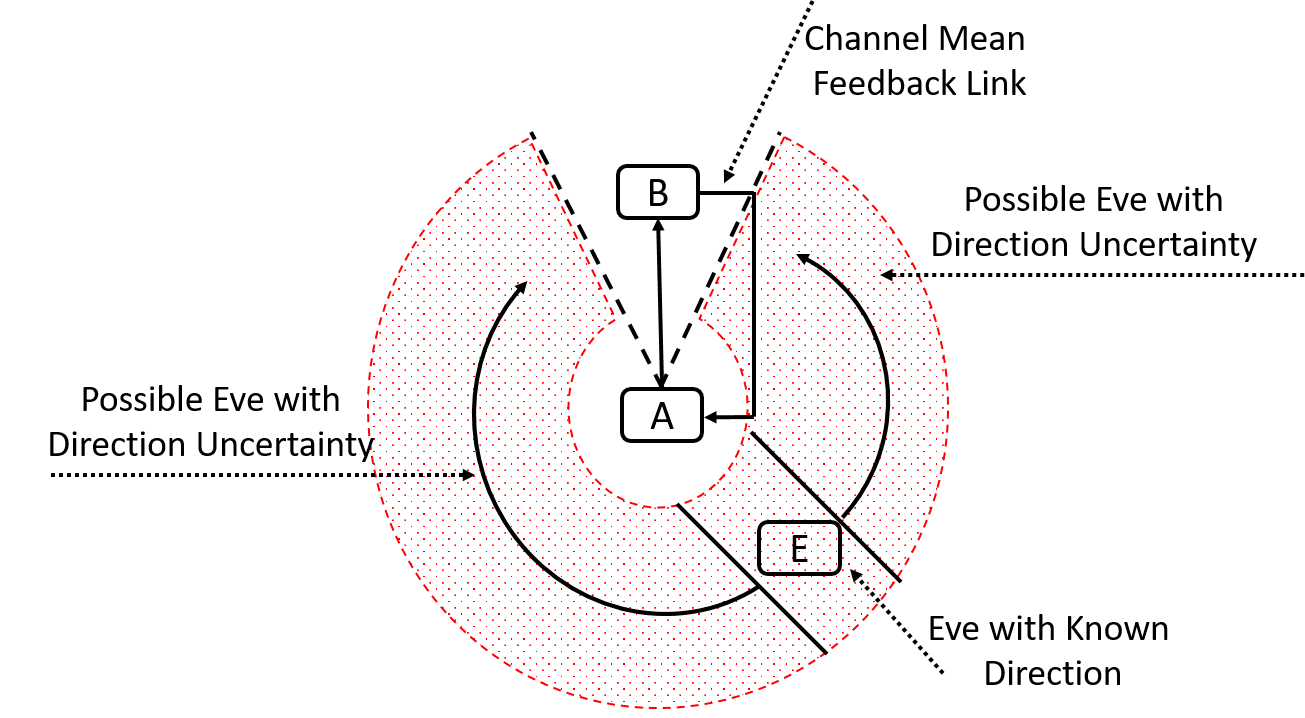}
\caption{Wiretap channel with physically constrained eavesdropper 
\label{fig:Compound_wiretap_model}}
\end{figure}
First, let us define
\begin{align}
\label{eq:seCap_known_phi}
C(\mathbf{W}_b,\mathbf{W}_e,\mathbf{Q}) = \log \dfrac{\det(\mathbf{I}_{N_a} + \mathbf{W}_b\mathbf{Q})}{\det(\mathbf{I}_{N_a} + \mathbf{W}_e\mathbf{Q})},
\end{align}
%and we have used the determinant identity $\det(\mathbf{I}_n + \mathbf{A}_{n\times m}\mathbf{B}_{m\times n}) = \det(\mathbf{I}_m + \mathbf{B}_{m\times n}\mathbf{A}_{n\times m})$. 
where $\mathbf{W}_{\circ} \triangleq \mathbf{H}_{\circ}^{\dagger}\mathbf{H}_{\circ}$, ${\circ} \in\{e,b\}$  is the channel Gram matrix and $\mathbf{Q}=\mathbb{E}\left[\mathbf{x}\mathbf{x}^{\dagger}\right]$ is the input signal covariance matrix. The capacity of the worst case main and eavesdropper channels can be defined as follows:
\begin{align}
\label{eq:CW}
C_w =   \min_{\substack{\mathbf{W}_b : \mathbf{H}_b\in \mathcal{S}_b \\ \mathbf{W}_e : \mathbf{H}_e\in \mathcal{S}_e}} \max_{\substack{\mathbf{Q} \succeq \mathbf{0} \\ \mathbf{tr}(\mathbf{Q}) \leq P}}  C(\mathbf{W}_b,\mathbf{W}_e,\mathbf{Q}). 
\end{align}
The following lower bound on the compound secrecy capacity was established in \cite{schaefer2015secrecy}:
\begin{align}
\label{eq:Cl}
C_l =    \max_{\substack{\mathbf{Q} \succeq \mathbf{0} \\ \mathbf{tr}(\mathbf{Q}) \leq P}} \min_{\substack{\mathbf{W}_b : \mathbf{H}_b\in \mathcal{S}_b \\ \mathbf{W}_e : \mathbf{H}_e\in \mathcal{S}_e}} C(\mathbf{W}_b,\mathbf{W}_e,\mathbf{Q}). 
\end{align}
thus, the following bounds on the compound capacity holds \cite{schaefer2015secrecy}:
\begin{align}
C_l \leq C_c \leq C_w
\end{align}
The problem under consideration is first to evaluate the lower bound on the compound capacity over the uncertainty sets by solving (\ref{eq:Cl}). To solve (\ref{eq:Cl}) we need to identify the worst case main and eavesdropper channels (i.e. main and eavesdropper channel realizations that minimize the lower bound), and then, determining the optimal signaling scheme accordingly. Further, we need to check whether the saddle point property, in the form $\max \min = \min \max$, for the considered class of channels.

%%%%%%%%%%%%%%%%%%%%%%%%%%%%%%%%%%%%%%%%%%%%%%%%%%%%%%%%%%%%%%%%%%%%%%%%%%%%%%%%%%%%
%%%%%%%%%%%%%%%%%%%%%%%%%%%%%%%%%%%%%%%%%%%%%%%%%%%%%%%%%%%%%%%%%%%%%%%%%%%%%%%%%%%%
%%%%%%%%%%%%%%%%%%%%%%%%%%%%%%%%%%%%%%%%%%%%%%%%%%%%%%%%%%%%%%%%%%%%%%%%%%%%%%%%%%%%
%%%%%%%%%%%%%%%%%%%%%%%%%%%%%%%%%%%%%%%%%%%%%%%%%%%%%%%%%%%%%%%%%%%%%%%%%%%%%%%%%%%%
\section{Compound Secrecy Capacity With Known Eavesdropper Mean Information}
\label{sec:known_location}
In this section we characterize the secrecy capacity of the considered compound wiretap channel when the mean information of the eavesdropper, $u_e$, is known exactly at the transmitter, i.e., $\mathcal{U}=\{u_e\}$. To proceed, we first need to identify the worst main channel Gram matrix, $\mathbf{W}_{bw}$ which is evaluated in the following proposition:
\begin{prop}
\label{prop:worst_Wb}
For the considered compound wiretap channel, for any non negative definite matrix $\mathbf{Q}$ and any $\mathbf{W}_e$ such that $ \mathbf{H}_e \in \mathcal{S}_e$ we have
\begin{align}
\label{eq:worst_Wb}
C(\mathbf{W}_b,\mathbf{W}_{e},\mathbf{Q}) \geq C(\mathbf{W}_{bw},\mathbf{W}_{e},\mathbf{Q})
\end{align}
where $\mathbf{W}_{bw} = (\lambda_{\mu b}^{1/2} - \epsilon_b)_+^{2} u_bu_b^{\dagger}$ is the worst main channel Gram matrix where $(x)_{+} = \max (0,x)$.
\end{prop}
\begin{proof}
The proof is give in Appendix \ref{app:W_bw}.
\end{proof}
The result of Proposition \ref{prop:worst_Wb} can interpreted as follows, the worst main channel is the channel that has lost all, but one, of its degrees of freedom, meanwhile, the only left degree of freedom happens with its maximum possible strength in the direction that is anti-parallel to the mean channel. An obvious direct consequence of Proposition \ref{prop:worst_Wb} is that the optimal input covariance, $\mathbf{Q}^*$, has to be of unit rank. That is because $\mathbf{W}_{bw}$ is shown to be of unit rank. Therefore, we conclude that beamforming is the optimal transmit strategy for the considered compound wiretap channel. Thus, we can restrict our analysis to a unit rank $\mathbf{Q}$. Next, we need to identify the worst eavesdropper Gram matrix.
\begin{prop}
\label{prop:worst_We}
For the considered compound wiretap channel, for any unit rank matrix $\mathbf{Q}$ with $\lambda(\mathbf{Q}) \leq P$, we have
\begin{align}
\label{eq:worst_We}
C(\mathbf{W}_{bw},\mathbf{W}_e,\mathbf{Q}) \geq C(\mathbf{W}_{bw},\mathbf{W}_{ew},\mathbf{Q})
\end{align}
where $\mathbf{W}_{ew} = (\lambda_{\mu e}+ 2 \lambda_{\mu e}^{1/2}\epsilon_e) u_e u_e^{\dagger} + \epsilon_e^2 \mathbf{I}$ is the worst eavesdropper channel Gram matrix.
\end{prop}
\begin{proof}
The proof is given in Appendix \ref{app:W_ew}.
\end{proof}

The statement of proposition \ref{prop:worst_We} states that the worst eavesdropper channel is isotropic around its mean channel. This means that, the worst eavesdropper dropper channel happens with its maximum strength in the direction parallel to its mean channel. The main result of this paper is given in the following theorem. We give the compound secrecy capacity of the considered class of channel and the capacity achieving input signal covariance $\mathbf{Q}^*$.
\begin{theorem}
\label{thm:Q*_known}
The secrecy capacity for the compound wiretap channel defined in (\ref{eq:S1}) and (\ref{eq:S2}) is equal to the worst case capacity, the saddle point property holds
\begin{align}
\label{eq:Cs_known}
C_c^* &= \max_{\substack{\mathbf{Q} \succeq \mathbf{0} \\ \mathbf{tr}(\mathbf{Q}) \leq P}} \min_{\substack{\mathbf{W}_b : \mathbf{H}_b\in \mathcal{S}_b \\ \mathbf{W}_e : \mathbf{H}_e\in \mathcal{S}_e}} C(\mathbf{W}_b,\mathbf{W}_e,\mathbf{Q}) \nonumber \\
&= \min_{\substack{\mathbf{W}_b : \mathbf{H}_b\in \mathcal{S}_b \\ \mathbf{W}_e : \mathbf{H}_e\in \mathcal{S}_e}} \max_{\substack{\mathbf{Q} \succeq \mathbf{0} \\ \mathbf{tr}(\mathbf{Q}) \leq P}} C(\mathbf{W}_b,\mathbf{W}_e,\mathbf{Q})\nonumber \\
&= C(\mathbf{W}_{bw},\mathbf{W}_{ew},\mathbf{Q}^*)= C_w
\end{align}
where $\mathbf{W}_{bw}$ and $\mathbf{W}_{ew}$ are as given in Propositions \ref{prop:worst_Wb} and \ref{prop:worst_We} respectively. Moreover, beamforming is the optimal signaling strategy:
\begin{align}
\label{eq:q*_known}
\mathbf{Q}^* = P q_* q_*^{\dagger}, 
\end{align}
where $q_*$ is the eigenvector associated with the maximum eigenvalue of $(\mathbf{I}_{Na} + P \mathbf{W}_{ew})^{-1}(\mathbf{I}_{Na} + P \mathbf{W}_{bw})$.
%, and 
%\begin{align}
%\label{eq:indic}
%\mathbbm{1}_{*} = \left\{
%	\begin{array}{ll}
%		1     & \mbox{if } (\lambda_{\mu b}^{1/2} - \epsilon_b)_+^{2}u_b^{\dagger}q_* >  \epsilon_e^2\\
%		0           & \mbox{if } (\lambda_{\mu b}^{1/2} - \epsilon_b)_+^{2}u_b^{\dagger}q_*  \leq  \epsilon_e^2 \
%	\end{array}
% \right. 
%\end{align}
\end{theorem}
\begin{proof}
The proof is given in Appendix \ref{app:Q_known}. 
\end{proof}
Theorem \ref{thm:Q*_known} proves that the saddle point property holds for the class of channels described in (\ref{eq:S1}) and (\ref{eq:S2}), and thus, the secrecy capacity of the compound wiretap channel is equal to the worst case capacity. Further, since the worst case main channel is of unit rank, accordingly, generalized eigenvector beamforming is known to be the optimal signaling strategy \cite{loyka2012optimal}\cite{li2009transmitter}.
% The indicator function, $\mathbbm{1}_*$, is used to control the transmission decision where it allows transmission only when positive secrecy rates can be achieved.

\subsection{Null Steering Beamforming as an Alternative Solution}
 As can be seen from Theorem \ref{thm:Q*_known}, the generalized eigenvector solution requires matrix inversion which may require a considerably high computational complexity especially when the number of transmitting antennas gets large. Therefore, we introduce the null steering (NS) beamforming \cite{null_Steering} as an alternative suboptimal solution. In our case, the NS beamforming matrix is given as follows:
\begin{align}
\label{eq:NS}
\mathbf{Q}_{ns} = P q_{ns} q_{ns}^{\dagger},\;\;\;\;\;\;& q_{ns}= \dfrac{\left[\mathbf{I}-u_eu_e^{\dagger}\right]u_b}{\norm{\left[\mathbf{I}-u_eu_e^{\dagger}\right] u_b}}.
\end{align}
NS beamformer can recognized as the projection of $u_b$ onto the null space of $u_w$. In particular, $q_{ns}$ maximizes the gain in the direction $u_b$ while creating a null notch in the direction $u_e$. Thus, it can understood as the transmission in the direction orthogonal to the eavesdropper mean channel direction maintaining the maximum possible gain in mean main channel direction. We give justifications for the choice of NS beamforming as a candidate suboptimal solution for our problem in Appendix \ref{app:NS}. Extensive computer simulation provided in section \ref{sec:simulation} reveals that NS performs extremely close to the optimal solution, yet, with no need for matrix inversion.

%%%%%%%%%%%%%%%%%%%%%%%%%%%%%%%%%%%%%%%%%%%%%%%%%%%%%%%%%%%%%%%%%%%%%%%%%%%%%%%%%%%%
%%%%%%%%%%%%%%%%%%%%%%%%%%%%%%%%%%%%%%%%%%%%%%%%%%%%%%%%%%%%%%%%%%%%%%%%%%%%%%%%%%%%
%%%%%%%%%%%%%%%%%%%%%%%%%%%%%%%%%%%%%%%%%%%%%%%%%%%%%%%%%%%%%%%%%%%%%%%%%%%%%%%%%%%%
%%%%%%%%%%%%%%%%%%%%%%%%%%%%%%%%%%%%%%%%%%%%%%%%%%%%%%%%%%%%%%%%%%%%%%%%%%%%%%%%%%%%
\section{Compound Secrecy Capacity With Eavesdropper Mean Uncertainty}
\label{sec:unknown_location}
Unlike the previous section where $u_e$ is assumed to be known at the transmitter, in this section we characterize the secrecy capacity of the considered compound wiretap channel when the eavesdropper mean direction, $u_e$, is known only to belong to the set $\mathcal{U}$.
% This scenario is much interesting from the CSI feedback perspective. While its unlikely to expect the eavesdropper to share its CSI, even its mean channel, to make its task harder, in some practical scenarios the nature of the communication environment may limit the eavesdropper access to the communication environment. This is likely to happen due to, e.g., natural obstacles (mountains, water surface, large buildings) or for the sake of physical security clearance around either transmitter or receiver or, may be, both. This limits the eavesdropper (if any) access to the communication environment. These physical restrictions can be informative to the transmitter and may enable to achieve better secure rate by designing its signaling scheme accordingly.
 A key step toward the characterization of the secrecy capacity is to find the worst eavesdropper channel Gram matrix $\mathbf{W}_{ew}$ with that assumption. We give $\mathbf{W}_{ew}$ in the following proposition.
\begin{prop}
\label{prop:worst_We_unknown}
For the considered compound wiretap channel with $u_e \in \mathcal{U}$, for all $\mathbf{W}_b \in \mathcal{S}_1$ and any non negative definite matrix $\mathbf{Q}$ we have
\begin{align}
\label{eq:worst_We_unknown}
C(\mathbf{W}_b,\mathbf{W}_e,\mathbf{Q}) \geq C(\mathbf{W}_b,\mathbf{W}_{ew},\mathbf{Q})
\end{align}
where $\mathbf{W}_{ew} = (\lambda_{\mu e}+ 2 \lambda_{\mu e}^{1/2}\epsilon_e) u_* u_*^{\dagger} + \epsilon_e^2 \mathbf{I}$ is the worst eavesdropper channel Gram matrix where:
\begin{align}
\label{eq:u*}
u_* = \argmax_{u \in \mathcal{U}} u_b^{\dagger}u,
\end{align}
\end{prop}
\begin{proof}
The proof is given in appendix \ref{app:W_e_un}.
\end{proof}
Observe that, the assumption that $u_e \in \mathcal{U}$ does not affect $\mathbf{W}_{bw}$, thus, the optimal covariance is again of unit rank as in the case $\mathcal{U}=\{u_e\}$. Therefore, beamforming is still the optimal transmit strategy under the assumption $u_e \in \mathcal{U}$. Again, transmission in the direction of the eigenvector of $(\mathbf{I}_{N_a} + P \mathbf{W}_{ew})^{-1}(\mathbf{I}_{N_a} + P \mathbf{W}_{bw})$ is the optimal solution where $\mathbf{W}_{ew}$ as given in proposition \ref{prop:worst_We_unknown}. We give the optimal $\mathbf{Q}^*$ in the following corollary as a direct consequence of Theorem \ref{thm:Q*_known}
\begin{cor}
\label{cor:Q*_unknown}
The saddle point property \eqref{eq:Cs_known} holds for the considered compound wiretap channel with $u_e \in \mathcal{U}$. 
%is given by
%\begin{align}
%\label{eq:Cs_unknown}
%C_c^* &= \max_{\substack{\mathbf{Q} \succeq \mathbf{0} \\ \mathbf{tr}(\mathbf{Q}) \leq P}} \min_{\substack{\mathbf{W}_b : \mathbf{H}_b\in \mathcal{S}_b \\ \mathbf{W}_e : \mathbf{H}_e\in \mathcal{S}_e}} C(\mathbf{W}_b,\mathbf{W}_e,\mathbf{Q}) \nonumber \\
%&=  \min_{\substack{\mathbf{W}_b : \mathbf{H}_b\in \mathcal{S}_b \\ \mathbf{W}_e : \mathbf{H}_e\in \mathcal{S}_e}} \max_{\substack{\mathbf{Q} \succeq \mathbf{0} \\ \mathbf{tr}(\mathbf{Q}) \leq P}} C(\mathbf{W}_b,\mathbf{W}_e,\mathbf{Q})\nonumber \\
%&= C(\mathbf{W}_{bw},\mathbf{W}_{ew},\mathbf{Q}^*)=C_w,
%\end{align}
 Moreover, the optimal signaling scheme is zero mean Gaussian with covariance matrix given by 
\begin{align}
\label{eq:q*_unknown}
\mathbf{Q}^* = P q_* q_*^{\dagger}, 
\end{align}
where $q_*$ is the eigenvector associated with the maximum eigenvalue of $(\mathbf{I}_{Na} + P \mathbf{W}_{ew})^{-1}(\mathbf{I}_{Na} + P \mathbf{W}_{bw})$, where $\mathbf{W}_{bw}$ and $\mathbf{W}_{ew}$ are as given in propositions \ref{prop:worst_Wb} and  \ref{prop:worst_We_unknown}, respectively.
%, and 
%\begin{align}
%\mathbbm{1}_{*} = \left\{
%	\begin{array}{ll}
%		1     & \mbox{if } (\lambda_{\mu b}^{1/2} - \epsilon_b)_+^{2}u_b^{\dagger}q >  \epsilon_e^2\\
%		0           & \mbox{if } (\lambda_{\mu b}^{1/2} - \epsilon_b)_+^{2}u_b^{\dagger}q  \leq  \epsilon_e^2 \
%	\end{array}
% \right. 
%\end{align}

\end{cor}
\begin{proof}
Follows immediately by Theorem \ref{thm:Q*_known} while realizing that the worst eavesdropper mean channel is in the direction $u_*$. 
\end{proof}  
Corollary \ref{cor:Q*_unknown} extends Theorem \ref{thm:Q*_known} to the case of uncertainty about the eavesdropper mean channel direction. We can conclude that, since the transmitter does not know the eavesdropper mean channel, it design its signal assuming the worst eavesdropper mean channel. Again, we note that, NS beamforming still can be introduced as an alternative solution against an eavesdropper with mean direction uncertainty. For this particular scenario, $q_{ns}$ takes the same form as in (\ref{eq:NS}), yet, in the direction $u_*$ instead of $u_e$.
\section{Application to Rician Fading MIMO Wiretap Channel}
In this section we consider a special class of MIMO wiretap channels which is well adopted to the class of compound wiretap channel considered in this paper. We study the Rician fading MIMO wiretap channel. In a Rician fading environment, the deterministic line of sight (LOS) component causes the channel variations to be centered around a mean matrix. This mean matrix is usually of unit rank whose gain depends mainly on the distance between transmitter and receiver, array configuration and respective array orientation. In the next section we give the Rician fading MIMO channel model, and then, in section \ref{sec:conn_los_compound} we describe the relation between Rician MIMO wiretap channel and the compound wiretap channel described in (\ref{eq:S1}) and (\ref{eq:S2}).  
\subsection{Rician Fading MIMO Channel Model} 
\label{sec:chan_model}
Wireless MIMO channel with dominant LOS component is best described by the Rician fading model. In Rician fading model, the received signal can be decomposed into two components; one is the specular component originated from the LOS path and the other is the diffuse non-line of sight component (NLOS) component. Following, we give the mathematical model for the considered Rician MIMO wiretap channel with the subscript $\circ \in \{b,e\}$ denotes the legitimate and eavesdropper channels respectively. 
\begin{align}
\label{eq:ric_decomp}
\mathbf{H}_{\circ} = \mathbf{H}_{\circ}^{los} + \mathbf{H}_{\circ}^{nlos}, 
\end{align}
where $\mathbf{H}_{\circ}^{los}$ and $\mathbf{H}_{\circ}^{nlos}$ represents the LOS and NLOS components respectively and
\begin{align}
\label{eq:ric_decomp_det}
\mathbf{H}_{\circ}^{los} &= \sqrt{\dfrac{\gamma_{\circ}^2 k_{\circ}}{1+k_{\circ}}} {\mathbf{\Psi}_{\circ}}, &
\mathbf{H}_{\circ}^{nlos}  &= \sqrt{\dfrac{\gamma_{\circ}^2}{1+k_{\circ}}}\hat{\mathbf{H}}_{\circ}, 
\end{align}
where $\gamma_{\circ}$ quantifies the channel strength for both receiver and eavesdropper, $k_{\circ}$ is the Rician factor that facilitates the contribution of the LOS component to the received signal, $\mathbf{\Psi}_{\circ} =\mathbf{a}(\theta_{\circ})\mathbf{a}^{\dagger}(\phi_{\circ})$, $\mathbf{a}(\theta_{\circ})$ and $\mathbf{a}(\phi_{\circ})$ are the antenna array spatial signatures (steering vectors) at receiver (eavesdropper) and transmitter respectively, $\theta_{\circ}$ and $\phi_{\circ}$ are the angle of arrival (AoA) and angle of departure (AoD) of the transmitted signal respectively. Note that, AoD, $\phi$, represents the \textit{azimuth angle} of the receiver (eavesdropper) with respect to the transmitter antenna array. Meanwhile, $\hat{\mathbf{H}}_{\circ}$ represents the channel coefficients matrix for the NLOS signal component.

\subsection{Relation to the Compound Wiretap Channel} 
\label{sec:conn_los_compound}
In the previous section we gave the mathematical description of the Rician fading MIMO wiretap channel. In this section we highlight the equivalence relation between this class of MIMO wiretap channel and the compound wiretap channel studied in this paper. Recalling the definition of the compound wiretap channel given in (\ref{eq:S1}) and (\ref{eq:S2}), it is straight forward to see that the following analogies hold:
\begin{align}
\label{eq:analogy}
\mathbf{H}_{\mu \circ}   &\Leftrightarrow  \mathbf{H}_{\circ}^{los}, & \lambda_{\mu \circ}      &\Leftrightarrow  \dfrac{N_aN_{\circ}\gamma_{\circ}^2 k_{\circ}}{1+k_{\circ}}, \nonumber \\
v_{\circ}                &\Leftrightarrow  \mathbf{a}(\theta_{\circ}), & u_{\circ}                &\Leftrightarrow  \mathbf{a}(\phi_{\circ}), \nonumber \\
\Delta\mathbf{H}_{\circ} &\Leftrightarrow  \mathbf{H}_{\circ}^{nlos}, & \epsilon_{\circ}^2         &\Leftrightarrow  \dfrac{N_{\circ}\gamma_{\circ}^2}{1+k_{\circ}}.
\end{align}   
Observe that in the settings of Rician fading MIMO wiretap channel, eavesdropper eigen direction, $u_e$, corresponds to the physical direction (in azimuth plane) of the eavesdropper. Therefore, the assumptions that $u_e$ is known at the transmitter corresponds to the scenario in which the transmitter has a prior knowledge about the eavesdropper azimuth direction. Whereas, the assumption that $u_e \in \mathcal{U}$ corresponds to the scenario in which the transmitter does not know exactly the azimuth direction of the eavesdropper, meanwhile, it knows that the eavesdropper, if any, has a restricted access to the communication area. That is, it can only approach the transmitter up to a certain distance and the receiver up to a certain azimuth direction.
\subsection{Numerical results}
\label{sec:simulation}
For the sake of numerical evaluation, we use the established equivalence relation between the considered compound wiretap channel and the MIMO channel with Rician fading. We compare the performance of the optimal solution to our proposed NS beamforming solution. Given the azimuth direction of both eavesdropper and legitimate receiver, another two possible transmission schemes may come to mind. First, beamforming toward the intended receiver which is well known as MRT. Second, creating a deep null notch in the direction of the eavesdropper which is well known as ZF. To evaluate the value of the mean information, we provide numerical simulation for an eavesdropper having the same parameters as of the main receiver, i.e. $N_a=N_b=N_e=4$ and $\gamma_b=\gamma_e=1$, for the same Rician $k$ factor, and thus, we have $\epsilon_b=\epsilon_e$ and $\lambda_{\mu b}=\lambda_{\mu e}$. We assume uniform linear array configuration at all nodes with antenna spacing of half wavelength. Whereas, we assume the receiver and eavesdropper not to share the same azimuth direction, $\phi_b = 25^{\circ}$ and $\phi_e = 60^{\circ}$. As can be seen in Fig. (\ref{fig:Secrecy_Capacity}), we compare the achievable secrecy rate for the NS, MRT and ZF beamforming approaches against the optimal solution for different values of Rician $k$ factor. Simulation results shows that NS beamforming performs extremely close to optimal and outperforms both MRT and ZF over the entire SNR range. Although It may be seen that NS performance matches the optimal solution, we provide a zoom in picture at the upper left corner of Fig. (\ref{fig:Secrecy_Capacity}) to show that the acheivable rate by NS beamforming is slightly below the secrecy capacity of the channel. It is observed that it maintains a small gap to capacity in order of $10^{-4}$ over the entire SNR range for all values of $k$.       
\begin{figure}[ht]
\centering
\includegraphics[width=3 in, height = 2 in]{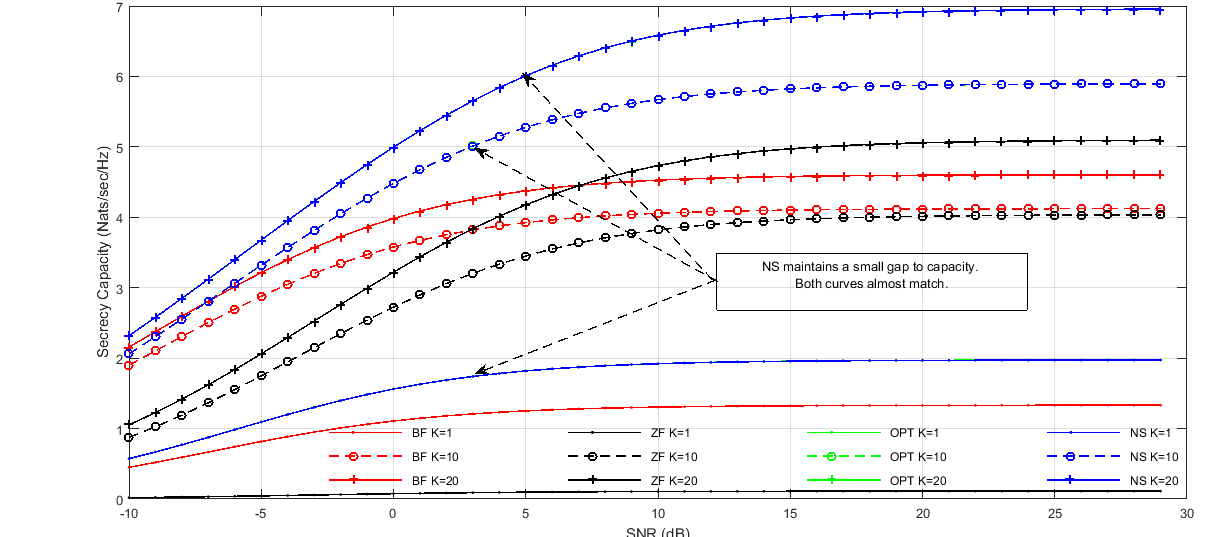}
\caption{Secrecy capacity of Rician compound wiretap channel with mean feedback. NS beamforming performs extremely close to optimal and outperforms both MRT and ZF approaches over the entire SNR range for different values of Rician $k$ factor. Legitimate transmitter is located at $25^{\circ}$, eavesdropper is at $60^{\circ}$. Number of antennas is $4$ for all entities.  
\label{fig:Secrecy_Capacity}}
\end{figure}
\section{Discussion and Future Work}
Compound MIMO wiretap channel with double sided uncertainty is considered under channel mean information model. The worst case main channel is shown to be anti-parallel to the channel mean information resulting in an overall unit rank channel. Further, the worst eavesdropper dropper channel is shown to be isotropic around its mean information. Accordingly, generalized eigenvector beamforming is shown to be the optimal signaling strategy. The saddle point property is shown to hold under mean information model, and thus, compound secrecy capacity equals to the worst case capacity over the class of uncertainty. Further, as the generalized eigenvector solution requires matrix inversion, we introduced NS beamforming, that is, transmission in the direction orthogonal to the eavesdropper mean channel direction maintaining the maximum possible gain in mean main channel direction, as an alternative suboptimal solution. Extensive computer simulation revealed that NS performs extremely close to the optimal solution. It also verified the superiority of NS beamforming to both MRT and ZF approaches over the entire SNR range.
\par 
It is worth noting that, the results for the compound wiretap channel are too conservative in general and, consequently, so is the result of this paper. That is due to the assumption that channel realizations remain constant over the entire transmission duration leading us to the worst case optimization. While this assumption simplifies the mathematical analysis, it does not usually hold in practice. More interesting scenario is to consider the compound wiretap channel with channel realizations allowed to change, possibly random, during transmission duration. 
\label{sec:conclusion}



\bibliographystyle{IEEEtran}
\bibliography{IEEEabrv,References}
% \end{thebibliography}
\begin{appendices}
\section{Proof of Proposition \ref{prop:worst_Wb}}
\label{app:W_bw}
We observe that
\begin{align}
 C(\mathbf{W}_b,\mathbf{W}_{e},\mathbf{Q}) &= \log \dfrac{\det(\mathbf{I}_{N_a} + \mathbf{W}_b\mathbf{Q})}{\det(\mathbf{I}_{N_a} + \mathbf{W}_{e}\mathbf{Q})} \nonumber \\
&\stackrel{(a)}{=} \sum_{i=1}^{N_a} \log(1 + \lambda_i(\mathbf{W}_b\mathbf{Q})) \nonumber \\
&- \log\det(\mathbf{I}_{N_a} + \mathbf{W}_{e}\mathbf{Q}) \nonumber \\
&\stackrel{(b)}{=} \sum_{i=1}^{N_a} \log(1 + \sigma_i^2((\mathbf{H}_{\mu b}+\Delta\mathbf{H}_b)\mathbf{Q}^{1/2})) \nonumber \\
&- \log\det(\mathbf{I}_{N_a} + \mathbf{W}_{e}\mathbf{Q}) \nonumber \\
&\stackrel{(c)}{\geq} \log(1 + (\lambda_{\mu b}^{1/2} - \epsilon_b)_{+}^{2}\lambda_1(\mathbf{Q})) \nonumber \\
&- \log\det(\mathbf{I}_{N_a} + \mathbf{W}_{e}\mathbf{Q}) \nonumber \\
&\stackrel{(d)}{=} C(\mathbf{W}_{bw},\mathbf{W}_{e},\mathbf{Q})
\end{align}
where $(a)$ follows from determinant properties and $(b)$ follows by recognizing that $\lambda_i(\mathbf{W}_b\mathbf{Q}) = \sigma_i^2((\mathbf{H}_{\mu b}+\Delta\mathbf{H}_b)\mathbf{Q}^{1/2})$ where $\sigma_i(\mathbf{A})$ is the $i^{th}$ singular value of $\mathbf{A}$. Meanwhile, $(c)$ follows from the singular value inequality in Lemma 7 in \cite{schaefer2015secrecy}, that is, $\sigma_i^2((\mathbf{H}_{\mu b}+\Delta\mathbf{H}_b)\mathbf{Q}^{1/2}) \geq (\sigma_i(\mathbf{H}_{\mu b}) - \sigma_1(\Delta\mathbf{H}_b))_+ \lambda_i(\mathbf{Q})$ and removed the summation due to the fact that $\sigma_1(\mathbf{H}_{\mu b}) = \lambda_{\mu b}^{1/2}$ and $\sigma_i(\mathbf{H}_{\mu b}) = 0 \;\;\forall i>1$.
\section{Proof of Proposition \ref{prop:worst_We}}
\label{app:W_ew}
It can be seen that
\begin{align}
\label{eq:W_ew_proof}
 C(\mathbf{W}_{bw},\mathbf{W}_{e},\mathbf{Q})  &\stackrel{(a)}{=} \log \dfrac{\det(\mathbf{I}_{N_a} + (\lambda_{\mu b}^{1/2} - \epsilon_b)_+^{2} u_bu_b^{\dagger}\mathbf{Q})}{\det(\mathbf{I}_{N_a} + \mathbf{W}_{e}\mathbf{Q})} \nonumber \\
 &\stackrel{(b)}{=} \log \dfrac{1 + \lambda((\lambda_{\mu b}^{1/2} - \epsilon_b)_+^{2} u_bu_b^{\dagger}\mathbf{Q})}{1 + \sigma^2((\mathbf{H}_{\mu e} + \Delta \mathbf{H}_e)\mathbf{Q}^{1/2})} \nonumber \\
 &\stackrel{(c)}{=} \log \dfrac{1 + \lambda((\lambda_{\mu b}^{1/2} - \epsilon_b)_+^{2} u_bu_b^{\dagger}\mathbf{Q})}{1 + \sigma^2(\mathbf{H}_{\mu e} \mathbf{Q}^{1/2} + \Delta \mathbf{H}_e \mathbf{Q}^{1/2})} \nonumber \\
  &\stackrel{(d)}{\geq} \log \dfrac{1 + \lambda((\lambda_{\mu b}^{1/2} - \epsilon_b)_+^{2} u_bu_b^{\dagger}\mathbf{Q})}{1 + (\sigma(\mathbf{H}_{\mu e} \mathbf{Q}^{1/2}) + \sigma (\Delta \mathbf{H}_e \mathbf{Q}^{1/2}))^2} 
\end{align}
%&\stackrel{(e)}{\geq} \log \dfrac{1 + \lambda((\lambda_{\mu b}^{1/2} - \epsilon_b)_+^{2} u_bu_b^{\dagger}\mathbf{Q})}{1 + (\sigma(\mathbf{H}_{\mu e} \mathbf{Q}^{1/2}) + \sigma(\Delta \mathbf{H}_e) \sigma(\mathbf{Q}^{1/2}))^2}
where $(a)$ follows by direct substitution in (\ref{eq:seCap_known_phi}) with $\mathbf{W}_{bw}$ given in proposition \ref{prop:worst_Wb}, $(b)$ follows since $\mathbf{Q}$ is of unit rank and that $\lambda(\mathbf{W}_b\mathbf{Q}) = \sigma^2((\mathbf{H}_{\mu b}+\Delta\mathbf{H}_b)\mathbf{Q}^{1/2})$ and $(c)$ is straightforward. Meanwhile, the upper bound in $(d)$ follows since $\sigma(A+B) \leq \sigma(A) + \sigma(B)$ for unit rank matrices $A$ and $B$ where the inequality holds with equality when $A$ and $B$ have the same singular vectors. Therefore, the inequality in $(d)$ established with equality if $\mathbf{H}_{\mu e}$  and $\Delta\mathbf{H}_{e}$ have the same singular vectors. Let us write $\mathbf{H}_{\mu e} = V\Sigma_{\mu e}U^{\dagger}$ where the first columns of $V$ and $U$ are $v_e$ and $u_e$ respectively, and $\Sigma_{\mu e} = diag\{\lambda_{\mu e}^{1/2},0,..,0\}$. Hence to establish $(d)$ with equality, $\Delta \mathbf{H}_e$ need to have $V$ and $U$ as its right and singular vectors respectively. Consequently, we can write $\mathbf{H}_e= V\Sigma_eU^{\dagger}$, and hence, $\mathbf{W}_e = U\Sigma_e^2U^{\dagger}$. But we have that $\Sigma_{\mu e} \preceq \epsilon_e \mathbf{I}$, accordingly, $\Sigma_e \preceq \Sigma_{ew}$ where $\Sigma_{ew} = diag\{\lambda_{\mu e}^{1/2}+\epsilon_e, \epsilon_e,...,\epsilon_e\}$. Noting that the function $\log \det(\mathbf{I}+\mathbf{W}\mathbf{Q})$ is monotonically increasing in $\mathbf{W}$, we conclude that $\mathbf{W}_{ew} = U \Sigma_{ew}^2U^{\dagger}$. However, we can write $\mathbf{W}_{ew}$ as $(\lambda_{\mu e}+ 2 \lambda_{\mu e}^{1/2}\epsilon_e) u_e u_e^{\dagger} + \epsilon_e^2 \mathbf{I}$ as required. 

\section{Proof of Theorem \ref{thm:Q*_known}}
\label{app:Q_known}

To establish the saddle point property, we give a proof similar to the one given in Theorem 6 in \cite{schaefer2015secrecy} while keeping in mind the difference between the compound wiretap channel defined there with the one defined in (\ref{eq:S1}) and (\ref{eq:S2}). Let $\mathbf{Q}^*$ be the optimal solution for the left hand side max-min problem, we observe that, to show the saddle point property in (\ref{eq:Cs_known}) is equivalent to show that \cite{zeidler1994nonlinear}:
\begin{align}
C(\mathbf{W}_{bw},\mathbf{W}_{ew},\mathbf{Q}) &\stackrel{(a)}{\leq} C(\mathbf{W}_{bw},\mathbf{W}_{ew},\mathbf{Q}^*) \nonumber \\
&\stackrel{(b)}{\leq} C(\mathbf{W}_{b},\mathbf{W}_{e},\mathbf{Q}^*),
\end{align}
where $\mathbf{W}_{ew}$ and $\mathbf{W}_{bw}$ are as given in propositions \ref{prop:worst_We} and \ref{prop:worst_Wb} respectively. Note that $(a)$ follows since $\mathbf{Q}^*$ is optimal for $\mathbf{W}_b = \mathbf{W}_{bw}$ and $\mathbf{W}_e = \mathbf{W}_{ew}$. Now we write:
\begin{align}
 C(\mathbf{W}_{bw},\mathbf{W}_{ew},\mathbf{Q}^*) &\stackrel{(a)}{=} \log(1 + \sigma_i^2((\lambda_{\mu b}^{1/2} - \epsilon_b)_+^{2} u_bu_b^{\dagger}\mathbf{Q}^{*1/2})) \nonumber \\
&\;\;\;\;- \log \det(\mathbf{I}+\mathbf{W}_{ew}\mathbf{Q}^*)  \nonumber \\
&\stackrel{(b)}{\leq}\sum_{i=1}^{N_a} \log(1 + \sigma_i^2((\mathbf{H}_{\mu b}+\Delta\mathbf{H}_b)\mathbf{Q}^{*1/2})) \nonumber \\
&\;\;\;\;-\log \det(\mathbf{I}+\mathbf{W}_{ew}\mathbf{Q}^*)  \nonumber \\
&\stackrel{(c)}{=} C(\mathbf{W}_{b},\mathbf{W}_{ew},\mathbf{Q}^*) \nonumber \\
&\stackrel{(d)}{\leq} C(\mathbf{W}_{b},\mathbf{W}_{e},\mathbf{Q}^*)
\end{align}
where $(a)$ follows by direct substitution by $\mathbf{W}_{bw}$, meanwhile, $(b)$ and $(c)$ follow from (\ref{eq:worst_Wb}) and we used (\ref{eq:worst_We}) to write $(d)$. Since the difference channel is, at most, of unit rank, then, beamforming toward the eigenvector associated with the largest eigenvalue of $(\mathbf{I}_{Na} + P \mathbf{W}_{ew})^{-1}(\mathbf{I}_{Na} + P \mathbf{W}_{bw})$ follows by corollary 1 in \cite{loyka2012optimal} and Theorem 6 in \cite{li2009transmitter}.
% The indicator $\mathbbm{1}_*$ in (\ref{eq:indic}) follows immediately from \ref{eq:Q*_proof_known}$(d)$.  
\section{Proof of Proposition \ref{prop:worst_We_unknown}}
\label{app:W_e_un}
Observe that $\mathbf{H}_{\mu e}^{\dagger}\mathbf{H}_{\mu e} = \lambda_{\mu e}uu^{\dagger}$ for some $u \in \mathcal{U}$. Thus, the result of proposition \ref{prop:worst_We_unknown} can be established in a fashion similar to the proof of proposition \ref{prop:worst_We}, however, by realizing that
\begin{align}
\min_{\mathbf{W}_e : \mathbf{H}_e \in \mathcal{S}_e} C(&\mathbf{W}_{bw},\mathbf{W}_{e}(u),\mathbf{Q}) = \nn  \\
 &\min_{u \in \mathcal{U}} C(\mathbf{W}_{bw},(\lambda_{\mu e}+ 2 \lambda_{\mu e}^{1/2}\epsilon_e) u u^{\dagger} + \epsilon_e^2 \mathbf{I},\mathbf{Q}). 
\end{align}
Thus, taking the minimum of (\ref{eq:W_ew_proof}) over $u$ and dropping the constraint $u \in \mathcal{U}$, the minimum is attained when $u=u_b$. Meanwhile, with the constraint into action, the minimum is attained at $u_* \in \mathcal{U}$ which has the minimum distance to $u_b$. Equivalently,
\begin{align}
u_* &= \argmin_{u \in \mathcal{U}} \norm{u_b-u} \nonumber \\
    &= \argmax_{u \in \mathcal{U}}  u_b^{\dagger} u
\end{align} 
which agree with (\ref{eq:u*}).

\section{Justification for NS Beamforming}
\label{app:NS}
To understand the motivation behind introducing NS beamforming as an alternative solution, we write $\mathbf{Q} = P qq^{\dagger}$. Now, it can be seen that:
\begin{align}
\label{eq:Q*_proof_known}
C_c^* &= \max_{\substack{\mathbf{Q} \succeq \mathbf{0} \\ \mathbf{tr}(\mathbf{Q}) \leq P}} C(\mathbf{W}_{bw},\mathbf{W}_{ew},\mathbf{Q}) \nonumber \\
&\stackrel{(a)}{=} \max_{\substack{\mathbf{Q} \succeq \mathbf{0} \\ \mathbf{tr}(\mathbf{Q}) \leq P}} \log \dfrac{\det(\mathbf{I}+\mathbf{W}_{bw}\mathbf{Q})}{\det(\mathbf{I}+\mathbf{W}_{ew}\mathbf{Q})}\nonumber \\
&\stackrel{(b)}{=} \max_{\substack{\mathbf{Q} \succeq \mathbf{0} \\ \mathbf{tr}(\mathbf{Q}) \leq P}} \log \dfrac{\det(\mathbf{I}+(\lambda_{\mu b}^{1/2} - \epsilon_b)_+^{2} u_bu_b^{\dagger}\mathbf{Q})}{\det(\mathbf{I}+(\lambda_{\mu e}+ 2 \lambda_{\mu e}^{1/2}\epsilon_e) u_e u_e^{\dagger} + \epsilon_e^2 \mathbf{I} )\mathbf{Q})} \nonumber \\
&\stackrel{(c)}{=} \max_{\substack{q \\ \norm{q}=1}} \log \dfrac{\det(\mathbf{I}+P(\lambda_{\mu b}^{1/2} - \epsilon_b)_+^{2} u_bu_b^{\dagger}qq^{\dagger})}{\det(\mathbf{I}+P((\lambda_{\mu e}+ 2 \lambda_{\mu e}^{1/2}\epsilon_e)u_eu_e^{\dagger}qq^{\dagger} + \epsilon_e^2qq^{\dagger} ))} \nonumber \\
&\stackrel{(d)}{=} \max_{\substack{q \\ \norm{q}=1}} \log \dfrac{1 +P(\lambda_{\mu b}^{1/2} - \epsilon_b)_+^{2} u_b^{\dagger}q}{1+P((\lambda_{\mu e}+ 2 \lambda_{\mu e}^{1/2}\epsilon_e)u_e^{\dagger}q + \epsilon_e^2))}
\end{align}
where $(a)$ follows by direct substitution with $\mathbf{W}_{bw}$ and $\mathbf{W}_{ew}$ in (\ref{eq:seCap_known_phi}), $(b)$ follows by substituting the values of $\mathbf{W}_{bw}$ and $\mathbf{W}_{ew}$. Meanwhile, in $(c)$ we used that $\mathbf{Q}$ is of unit rank and thus it has only one eigenvalue equals to $P$ and its corresponding eigenvector $q$, also, we have removed the power constraint by introducing the constraint $\norm{q}=1$. Since both of the numerator and the denominator are of unit rank, $d$ follows from $(c)$ by substituting the only eigenvalue of both of them.  Now observe that, the choice of $q$ does not affect the eigenvalue of the matrix $\epsilon_e^2qq^{\dagger}$, rather, it do affect the eigenvalues of the other matrices. Clearly our objective is to find $q$ that simultaneously maximizes the numerator and minimizes (optimally, nulling out) the denominator in (\ref{eq:Q*_proof_known}$(d)$).  The optimal $q_*$ that maximizes $C_c$ is given by Theorem \ref{thm:Q*_known}. However, we note that $q_{ns}$ in (\ref{eq:NS}) is the optimal solution to the following optimization problem
\begin{align}
\max_{\substack{{q} \\ \norm{q}=1}} <q^{\dagger},u_b> \nonumber \\
\text{Subject to} <q^{\dagger},u_e> = 0,
\end{align} 
i.e., beamforming in the direction $q_{ns}$ maximizes the gain in the direction $u_b$ while creating a null notch in the direction $u_e$.
\end{appendices}
\end{document}